\documentclass{llncs}
\usepackage[utf8]{inputenc}
\usepackage{graphicx}
 
\usepackage{times}
\usepackage{soul}
\usepackage{url}
\usepackage[hidelinks]{hyperref}
\usepackage[utf8]{inputenc}
\usepackage[small]{caption}
\usepackage{graphicx}
\usepackage{amsmath}
\usepackage{booktabs}
\usepackage{bm}
\usepackage{stmaryrd}
\usepackage{lipsum}

\usepackage{stackrel}
\usepackage{amssymb}
\usepackage{amsmath}
\usepackage{mathrsfs}
\usepackage{comment}
\usepackage{todonotes}
\usepackage[inline]{enumitem}
\usepackage{float}
\usepackage{algorithmic}
\usepackage[linesnumbered,lined,commentsnumbered]{algorithm2e}
\usepackage{amsmath, commath}

\newtheorem{innercustomgeneric}{\customgenericname}
\providecommand{\customgenericname}{}
\newcommand{\newcustomtheorem}[2]{%
  \newenvironment{#1}[1]
  {%
   \renewcommand\customgenericname{#2}%
   \renewcommand\theinnercustomgeneric{##1}%
   \innercustomgeneric
  }
  {\endinnercustomgeneric}
}

\newcustomtheorem{customthm}{Theorem}
\newcustomtheorem{customlemma}{Lemma}
\usetikzlibrary{arrows,shapes}

\usetikzlibrary{fit, calc,positioning, decorations.pathreplacing,decorations.markings}

\tikzset{
  on each segment/.style={
    decorate,
    decoration={
      show path construction,
      moveto code={},
      lineto code={
        \path [#1]
        (\tikzinputsegmentfirst) -- (\tikzinputsegmentlast);
      },
      curveto code={
        \path [#1] (\tikzinputsegmentfirst)
        .. controls
        (\tikzinputsegmentsupporta) and (\tikzinputsegmentsupportb)
        ..
        (\tikzinputsegmentlast);
      },
      closepath code={
        \path [#1]
        (\tikzinputsegmentfirst) -- (\tikzinputsegmentlast);
      },
    },
  },
  mid arrow/.style={postaction={decorate,decoration={
        markings,
        mark=at position .5 with {\arrow[#1]{stealth}}
      }}},
}
\tikzset{
  big dot/.style={
    circle, inner sep=0pt, 
    minimum size=2mm, fill=black
  }
}
\tikzset{
  small dot/.style={
    circle, inner sep=0pt, 
    minimum size=1mm, fill=black
  }
}

\mathchardef\mhyphen="2D

\newlist{enumaa}{enumerate*}{5}
\setlist[enumaa]{label={(\emph{\alph*})}}

\newcommand{\leftrarrows}{\mathrel{\raise.75ex\hbox{\oalign{%
  $\scriptstyle\leftarrow$\cr
  \vrule width0pt height.5ex$\hfil\scriptstyle\relbar$\cr}}}}
\newcommand{\lrightarrows}{\mathrel{\raise.75ex\hbox{\oalign{%
  $\scriptstyle\relbar$\hfil\cr
  $\scriptstyle\vrule width0pt height.5ex\smash\rightarrow$\cr}}}}
\newcommand{\Rrelbar}{\mathrel{\raise.75ex\hbox{\oalign{%
  $\scriptstyle\relbar$\cr
  \vrule width0pt height.5ex$\scriptstyle\relbar$}}}}
\newcommand{\longleftrightarrows}{\leftrarrows\joinrel\Rrelbar\joinrel\lrightarrows}

\title{Category-theoretical Semantics of the Description Logic $\mathcal{ALC}$ (extended version)}
 
\author{Chan Le Duc}

\institute{Universit\'e Sorbonne Paris Nord, LIMICS, INSERM, U1142, F-93000, Bobigny, France
\\
\email{chan.leduc@univ-paris13.fr}
}

\begin{document}
\maketitle

\begin{abstract}
Category theory can be used to state formulas in First-Order Logic without using set membership. Several notable results in logic such as proof of the continuum hypothesis 
can be elegantly rewritten  in    category theory. We propose in  this paper a reformulation of the usual set-theoretical semantics of the description logic  $\mathcal{ALC}$ by using categorical language. In this setting, $\mathcal{ALC}$ concepts are represented as objects, concept subsumptions as  arrows, and memberships as logical quantifiers over objects and arrows of  categories. 
Such a category-theore\-tical semantics provides a more modular representation of the semantics of $\mathcal{ALC}$ and a new way to design algorithms for reasoning. \footnote{Copyright \textcopyright  2021 for this paper by its authors. Use permitted under
     Creative Commons License Attribution 4.0 International (CC BY 4.0)}
\end{abstract}

\section{Introduction} \label{sec:intro}

Languages based on Description Logics (DLs) \cite{baa10} such as OWL \cite{pat04}, OWL2 \cite{grau08}, are widely used to represent  ontologies   in semantics-based applications. $\mathcal{ALC}$ is the smallest DL involving roles which is closed under negation. It is a suitable logic for a first attempt to replace the usual set-theoretical semantics by another one.  A pioneer work by Lawvere \cite{law64} provided  an appropriate axiomatization of the category of sets by replacing set membership with the composition of functions. However, it was not indicated whether the categorial axioms are ``semantically"  equivalent to the axioms based on set membership. As pointed out by Goldblatt \cite{gol06}, this may lead to a very different semantics for negation.    

There have been very few works on connections between category theory and DLs. Spivak et al. \cite{spivak2012} used   category theory to define a high-level and graphical language comparable with OWL for knowledge representation rather than a foundational formalism for reasoning. 
In this paper,  we propose a rewriting of  the usual set-theoretical  semantics of $\mathcal{ALC}$ by using objects and arrows, which are  two fundamental  elements of category theory, to represent concepts and subsumptions respectively. 

\begin{example}\label{ex:intro}
We would like to model 5 states of a meeting: $\mathsf{arrived}$, $\mathsf{filled\mhyphen  room}$,  
$\mathsf{starting}$,  $\mathsf{started}$,  $\mathsf{finished}$; 
and  a chronological order between these states as depicted in Figure~\ref{fig:example}.
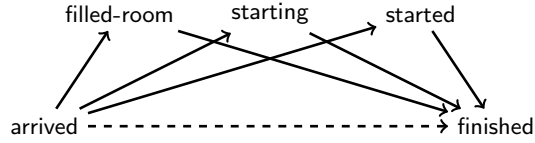
\begin{figure}[!htbp]
\centering
\begin{tikzpicture}[line width=1pt] 

\node[] (I) at (-3,-1) [label={[xshift=-0.3cm, yshift=-0.6cm]}] {$\mathsf{arrived}$};
\node[] (T) at (3,-1) [label={[xshift=-0.3cm, yshift=-0.6cm]}] {$\mathsf{finished}$};
\node[] (X) at (-2,0.5) [label={[xshift=-0.3cm, yshift=-0.6cm]}] {$\mathsf{filled\mhyphen room}$};
\node[] (Y) at (0,0.5) [label={[xshift=-0.3cm, yshift=-0.6cm]}] {$\mathsf{starting}$};
\node[] (Z) at (2,0.5) [label={[xshift=-0.3cm, yshift=-0.6cm]}] {$\mathsf{started}$};

 \draw[->] (I) to[]node[left=2pt]{ } (X);
 \draw[->] (X) to[]node[right=2pt]{ } (T);
 \draw[->] (Y) to[]node[right=2pt]{ } (T);
 \draw[->] (Z) to[]node[right=2pt]{ } (T);
 \draw[->] (I) to[]node[left=2pt]{ } (Y); 
 \draw[->] (I) to[]node[left=2pt]{ } (Z); 
 \draw[->,dashed] (I) to[]node[left=2pt]{ } (T); 
 
\end{tikzpicture}
\caption{Objects and arrows for a meeting}\label{fig:example}
\end{figure}
 If we use concepts and subsumptions in $\mathcal{ALC}$ to represent the states and order under set-theoretical semantics, then each such a concept should represent a set of time-points. This might not be compatible with the semantics of subsumptions with set membership since, for instance a time-point in $\mathsf{filled\mhyphen room}$ cannot belong to $\mathsf{finished}$ (at least this is not our modelling intention). If we use roles to model this chronological order, then transitive roles which   are not allowed in  $\mathcal{ALC}$,  should be needed to infer the dashed relationship. However, objects and arrows under category-theoretical semantics can be straightforwardly used to present such states and chronological order without referring to set membership. Moreover, transitive relationships such as the dashed one can be  inferred under category-theoretical semantics.
\end{example}

The paper is organized as follows. We  begin by translating  semantic constraints related to each $\mathcal{ALC}$ constructor into arrows between objects. Then, we check whether the obtained arrows allow to restore usual properties. If it is not the case, we add  new arrows and objects for capturing the missing properties without going beyond the set-theoretical semantics. For instance, it does not suffice to define category-theoretical semantics of negation $\neg C$ by stating $C \sqcap \neg C \longrightarrow \bot$ and $\top \longrightarrow C \sqcup \neg C$ because it is not possible to obtain  arrows such as  $C \longleftrightarrows \neg \neg C$ from this definition. After providing category-theoretical semantics of each constructor, unsatisfiability of  an $\mathcal{ALC}$ concept $C$ with respect to an ontology is reducible to the existence of a category consisting of an arrow $C\longrightarrow \bot$.   As a main result of the paper, we show that an $\mathcal{ALC}$ concept is       category-theoretically unsatisfiable iff it is set-theoretically unsatisfiable. This result allows us not only to employ  objects and arrows to design reasoning procedures in the new setting, but also to identify interesting sublogics by removing/weakening some arrows from categorial definitions. 
 
\section{Syntax and set-theoretical semantics of $\mathcal{ALC}$} \label{sec:alc}
We present syntax and semantics of the Description Logic $\mathcal{ALC}$ \cite{baa10} with  TBoxes and some basic inference  problems.

\begin{definition}[Syntax and set-theoretical  semantics]\label{def:alc}
Let  $\mathbf{R}$ and  $\mathbf{C}$   be non-empty sets of \emph{role names} and \emph{concept names}    respectively. 
The set of $\mathcal{ALC}$-concepts is inductively defined as the smallest set containing all concept names  in $\mathbf{C}$ with  $\top$, $\bot$ and complex concepts that are inductively defined as follows: $C\sqcap D$, $C\sqcup D$, $\neg C$, $\exists R.C$, $\forall R.C$ where  $C$ and $D$ are $\mathcal{ALC}$-concepts, and $R$ is a role name in $\mathbf{R}$. An axiom $C\sqsubseteq D$ is called a general concept inclusion (GCI) where $C,D$ are (possibly complex) $\mathcal{ALC}$-concepts. An $\mathcal{ALC}$ ontology  $\mathcal{O}$ is a finite set of GCIs. 
 
An interpretation ${\cal I}=\langle\Delta^{\cal{I}},\cdot^{\cal{I}}\rangle$ consists of a non-empty set $\Delta^{\cal{I}}$ ({\em domain}), and a function $\cdot^{\cal{I}}$ ({\em interpretation function}) which  associates a subset of $\Delta^{\cal{I}}$ to each concept name, an element  in $\Delta^\mathcal{I}$ to each individual, and a subset of $\Delta^{\cal{I}}\times \Delta^{\cal{I}}$ to each role name such that  
\begin{center}
\begin{tabular}{@{}l@{}ll@{}} 
$\top^{\cal{I}}$ & $=$ & $\Delta^{\cal{I}}$, $\bot^{\cal I} = \emptyset$, $(\neg C)^{\cal I}  =\Delta^{\cal{I}}\setminus C^{\cal I}$\\ 
$~(C\sqcap D)^{\cal I}$ &$=$ & $C^{\cal I}\cap D^{\cal I}$, $(C\sqcup D)^{\cal I}=C^{\cal I}\cup D^{\cal I}$\\
$(\exists R.C)^{\cal I}$ &$=$ & $\{x\!\in\!\Delta^{\cal{I}}\!\mid\!\exists y\!\in\!\Delta^{\cal{I}},\!\langle x,y\rangle\!\in\!{R^{\cal I}}\!\wedge\!y\!\in\! C^{\cal I}\}$\\
$(\forall R.C)^{\cal I}$ &$=$ & $\{x\!\in\!\Delta^{\cal{I}}\!\mid\! \forall y\!\in\!\Delta^{\cal{I}},\langle x,y\rangle\!\in\! {R^{\cal I}}\!\implies \! y\!\in\! C^{\cal I}\}$\\
\end{tabular}  
\end{center}

An interpretation $\mathcal{I}$ satisfies  a GCI $C\sqsubseteq D$  if $C^{\cal I}\subseteq D^{\cal I}$.
$\mathcal{I}$ is a model of  $\mathcal{O}$, written $\mathcal{I}\models  \mathcal{O}$, if   $\mathcal{I}$ satisfies  each GCI in $\mathcal{O}$.   
 
In this case, we say that $\mathcal{O}$ is {\em set-theoretically consistent}, and {\em set-theoretically inconsistent} otherwise. A concept $C$ is {\em set-theoretically satisfiable} with respect to $\mathcal{O}$ if there is a model $\mathcal{I}$ of $\mathcal{O}$ such that $C^\mathcal{I}\neq \emptyset$, and {\em set-theoretically unsatisfiable} otherwise. We say that a GCI $C\sqsubseteq D$ is {\em set-theoretically entailed} by $\mathcal{O}$, written $\mathcal{O}\models C\sqsubseteq D$, if $C^\mathcal{I}\subseteq D^\mathcal{I}$ for all models $\mathcal{I}$ of $\mathcal{O}$. The pair $\langle\mathbf{C},  \mathbf{R}\rangle$  is called the signature of $\mathcal{O}$.
\end{definition}

\section{Category-theoretical semantics of  $\mathcal{ALC}$} \label{sec:cat-semantics}

We can observe that the set-theoretical semantics of $\mathcal{ALC}$ is based on  set membership relationships   while ontology inferences such as consistency, concept subsumption involve  set inclusions. This explains why  inference algorithms  developed in this setting  build often a set of individuals connected  in some way for representing a model.  

In this section, we use some basic notions in category theory to characterize the semantics of $\mathcal{ALC}$. Instead of set membership, we use in this categorical language objects and arrows to define  semantics of a given object. Although the present paper is self-contained, we refer the readers to textbooks   \cite{gol06,saunders92} on category theory for further information.

\begin{definition}[Syntax categories]\label{def:syntax-cat}
Let  $\mathbf{R}$ and  $\mathbf{C}$   be non-empty sets of \emph{role names} and \emph{concept names}    respectively. We define a \emph{concept syntax category} $\mathscr{C}_c$ and a \emph{role syntax category} $\mathscr{C}_r$ from the signature $\langle\mathbf{C},  \mathbf{R}\rangle$  as follows:
\begin{enumerate}
   
\item Each role name $R$ in  $\mathbf{R}$ is an object $R$  of   $\mathscr{C}_r$. In particular, there are initial and terminal objects $R_\bot$ and $R_\top$ in     $\mathscr{C}_r$ with arrows $R\longrightarrow R_\top$ and $R_\bot\longrightarrow R$   for all object $R$ of  $\mathscr{C}_r$. There is also an identity arrow $R\longrightarrow R$ for each object $R$  of  $\mathscr{C}_r$.
 
    \item Each concept name in $\mathbf{C}$  is an object of $\mathscr{C}_c$. In particular, $\bot$ and $\top$  are respectively initial and terminal objects, i.e. there are arrows $C\longrightarrow \top$ and $\bot \longrightarrow C$ for each object $C$ of  $\mathscr{C}_c$. Furthermore, for each object  $C$  of  $\mathscr{C}_c$   there is an identity arrow $C\longrightarrow C$, 
    and for each object  $R$  of  $\mathscr{C}_r$ there  are two objects of $\mathscr{C}_c$, namely $\mathsf{dom}(R)$ and $\mathsf{cod}(R)$.

\item If there are arrows $E\longrightarrow F$ and $F\longrightarrow G$ in $\mathscr{C}_c$ (resp. $\mathscr{C}_r$), then there is an arrow $E\longrightarrow G$ in $\mathscr{C}_c$ (resp. $\mathscr{C}_r$).

\item There are two functors $\mathsf{dom}$ and $\mathsf{cod}$ from  $\mathscr{C}_r$ to $\mathscr{C}_c$, i.e.  they associate two objects $\mathsf{dom}(R)$ and $\mathsf{cod}(R)$ of $\mathscr{C}_c$ to each object $R$ of $\mathscr{C}_r$ such that  
\begin{enumerate}
\item $\mathsf{dom}(R_\top)=\top$, $\mathsf{cod}(R_\top)=\top$, $\mathsf{dom}(R_\bot)=\bot$ and $\mathsf{cod}(R_\bot)=\bot$.
    \item 
if there is an arrow $R\longrightarrow R$ in $\mathscr{C}_r$ then there are arrows $\mathsf{dom}(R)\longrightarrow \mathsf{dom}(R)$ and $\mathsf{cod}(R)\longrightarrow \mathsf{cod}(R)$.
\item if there are arrows $R\longrightarrow R'\longrightarrow R''$ in $\mathscr{C}_r$ then there are arrows $\mathsf{dom}(R)\longrightarrow \mathsf{dom}(R'')$ and $\mathsf{cod}(R)\longrightarrow \mathsf{cod}(R'')$.

\item if there is an arrow $\mathsf{dom}(R)\longrightarrow \bot$ or  $\mathsf{cod}(R)\longrightarrow \bot$ in $\mathscr{C}_c$, then there is  an arrow $R\longrightarrow R_\bot$ in $\mathscr{C}_r$. 
\end{enumerate}

\end{enumerate}
For each arrow $E\longrightarrow F$ in $\mathscr{C}_c$ or $\mathscr{C}_r$, $E$ and $F$ are respectively  called \emph{domain} and \emph{codomain} of the  arrow. 

\end{definition}

Definition~\ref{def:onto-cat} provides a general framework with syntax elements and necessary properties coming from category theory. We need to ``instantiate" it to obtain categories which include semantic constraints  coming from axioms. The following definition extends syntax categories in such a way that they admit the axioms of an $\mathcal{ALC}$  ontology as  arrows.
 
\begin{definition}[Ontology categories]\label{def:onto-cat}
Let $C_0$ be an $\mathcal{ALC}$ concept and $\mathcal{O}$ an $\mathcal{ALC}$ ontology from a  signature $\langle \mathbf{C}, \mathbf{R}\rangle$.  We define a \emph{concept ontology category} $\mathscr{C}_c\langle C_0, \mathcal{O}\rangle$ and a \emph{role ontology category} $\mathscr{C}_r\langle C_0, \mathcal{O}\rangle$  as follows:
\begin{enumerate}

\item $\mathscr{C}_c\langle C_0, \mathcal{O}\rangle$ and  $\mathscr{C}_r\langle C_0, \mathcal{O}\rangle$ are syntax categories from  $\langle \mathbf{C}, \mathbf{R}\rangle$.

\item $C_0$ is an object of $\mathscr{C}_c\langle C_0, \mathcal{O}\rangle$.

\item If $E \sqsubseteq F$ is an axiom of $\mathcal{O}$, then $\neg E \sqcup F$ is an object of  $\mathscr{C}_c\langle C_0, \mathcal{O}\rangle$ and there is an arrow $\top \longrightarrow \neg E \sqcup  F$  in $\mathscr{C}_c\langle C_0, \mathcal{O}\rangle$.

\end{enumerate}
\end{definition}

In this paper, an object of $\mathscr{C}_c\langle C_0, \mathcal{O}\rangle$ and $\mathscr{C}_r\langle C_0, \mathcal{O}\rangle$ is called concept and role object respectively. We transfer  the vocabulary used in Description Logics to  categories as follows. A concept object $C\sqcup D$, $C\sqcap D$ or $\neg C$ is respevtively called disjunction, conjunction and negation object. For an existential  restriction object $\exists R.C$ or universal  restriction object $\forall R.C$, $C$ is called the \emph{filler} of $\exists R.C$ and $\forall R.C$.

In the sequel,  we introduce category-theoretical semantics of  disjunction,  conjunction, negation, existential and universal restriction objects if they appear in $\mathscr{C}_c\langle C_0, \mathcal{O}\rangle$.
 Some of them require more \emph{explicit} properties than those needed for the set-theoretical semantics. This is due to the fact that set membership is translated into arrows in a syntax category.  Since semantics of an object in a category depends to relationships with another ones, the following definitions need to add to $\mathscr{C}_c\langle C_0, \mathcal{O}\rangle$ and $\mathscr{C}_r\langle C_0, \mathcal{O}\rangle$ new objects and arrows. 
 
 \begin{definition}[Category-theoretical semantics of  disjunction]\label{def:disj}
Let $C, D, C \sqcup D$ be   concept objects of $\mathscr{C}_c$.  Category-theoretical semantics of $C \sqcup D$ is defined by using arrows in $\mathscr{C}_c$ as follows. There are arrows $i,j$ from $C$ and $D$ to $C\sqcup D$, and if there is an object $X$ and arrows $f,g$ from $C,D$ to $X$, then there is  an arrow $k$ such that  the following diagram commutes  :

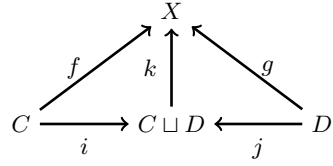
\begin{figure}[!htbp]
\centering
\begin{tikzpicture}[line width=1pt] 
\node[] (AB) at (0,-1) [label={[xshift=-0.3cm, yshift=-0.6cm]}] {$C \sqcup D$};
\node[] (A) at (-2,-1) [label={[xshift=-0.3cm, yshift=-0.6cm]}] {$C$};
\node[] (B) at (2,-1) [label={[xshift=-0.3cm, yshift=-0.6cm]}] {$D$};
\node[] (X) at (0,0.5) [label={[xshift=-0.3cm, yshift=-0.6cm]}] {$X$};

 \draw[->] (A) to[]node[below=2pt]{${i}$} (AB);
 \draw[->] (B) to[]node[below=2pt]{${j}$} (AB);
 \draw[->] (A) to[]node[left=2pt]{${f}$} (X);
 \draw[->] (B) to[]node[right=2pt]{${g}$} (X);
 \draw[->] (AB) to[]node[left=2pt]{${k}$} (X); 
\end{tikzpicture}
\caption{Commutative diagram for disjunction}\label{fig:disj}
\end{figure}

The diagram in Figure~\ref{fig:disj} can be rephrased  as follows:
\begin{align}
&C \longrightarrow C \sqcup D \text{ and } D  \longrightarrow  C \sqcup D   \label{disj01} \\
& \forall X, C    \longrightarrow  X \text{ and } D \longrightarrow  X  \Longrightarrow  C \sqcup D  \longrightarrow X \label{disj02}  
\end{align} 
\end{definition}
 
Intuitively speaking, Arrows~(\ref{disj01}) and (\ref{disj02}) tell us that  $C\sqcup D$ is the ``smallest" object which is ``greater" than $C$ and $D$.   

\begin{lemma}\label{lem:disj} The category-theorectical semantics of $C\sqcup D$ characterized by  Definition~\ref{def:disj} is compatible with the set-theoretical semantics of $C\sqcup D$, that means
if  $\langle \Delta^\mathcal{I}, \cdot^\mathcal{I}\rangle$ is an interpretation   under  set-theoretical semantics, then the following holds:

$(C\sqcup D)^\mathcal{I}=C^\mathcal{I}\cup D^\mathcal{I}$  iff
\begin{align}
&C^\mathcal{I} \subseteq (C  \sqcup D)^\mathcal{I  }, D^\mathcal{I}  \subseteq  (C  \sqcup D)^\mathcal{I}  \label{disj1}   \\
& \forall X\subseteq \Delta^\mathcal{I}, C^\mathcal{I}    \subseteq  X,   D^\mathcal{I}  \subseteq  X  \Longrightarrow  (C \sqcup D)^\mathcal{I}  \subseteq X   \label{disj2}  
\end{align}
\end{lemma}
\begin{proof}
 \noindent ``$\Longleftarrow$". Due to (\ref{disj1}) we have $C^\mathcal{I}\cup D^\mathcal{I} \subseteq (C\sqcup D)^\mathcal{I}$. Let   $X=C^\mathcal{I}\cup D^\mathcal{I}$. Due to (\ref{disj2}) we have  $(C\sqcup D)^\mathcal{I} \subseteq X=C^\mathcal{I}\cup D^\mathcal{I}$. 
 
 \noindent ``$\Longrightarrow$".  From $(C\sqcup D)^\mathcal{I}=C^\mathcal{I}\cup D^\mathcal{I}$, we have (\ref{disj1}).  Let $x\in (C\sqcup D)^\mathcal{I}$. Due to $(C\sqcup D)^\mathcal{I}=C^\mathcal{I}\cup D^\mathcal{I}$, we have  $x\in C^\mathcal{I}$ or $x\in D^\mathcal{I}$. Hence, $x\in X$  since $C^\mathcal{I}    \subseteq  X$ and  $D^\mathcal{I}  \subseteq  X$.\hfill$\square$
\end{proof}
 
At first glance, one can follow the same idea used in Definition~\ref{def:disj}   to define  category-theoretical semantics of $C\sqcap D$ as described in Figure~\ref{fig:conj}. They tell us that   $C\sqcap D$ is the ``greatest" object which is ``smaller" than $C$ and $D$.

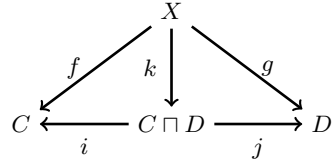
\begin{figure}[!htbp]
\centering
\begin{tikzpicture}[line width=1pt] 
\node[] (AB) at (0,-1) [label={[xshift=-0.3cm, yshift=-0.6cm]}] {$C \sqcap D$};
\node[] (A) at (-2,-1) [label={[xshift=-0.3cm, yshift=-0.6cm]}] {$C$};
\node[] (B) at (2,-1) [label={[xshift=-0.3cm, yshift=-0.6cm]}] {$D$};
\node[] (X) at (0,0.5) [label={[xshift=-0.3cm, yshift=-0.6cm]}] {$X$};

 \draw[->] (AB) to[]node[below=2pt]{${i}$} (A);
 \draw[->] (AB) to[]node[below=2pt]{${j}$} (B);
 \draw[->] (X) to[]node[left=2pt]{${f}$} (A);
 \draw[->] (X) to[]node[right=2pt]{${g}$} (B);
 \draw[->] (X) to[]node[left=2pt]{${k}$} (AB); 
\end{tikzpicture}
\caption{Commutative diagram for a  weak conjunction}\label{fig:conj}
\end{figure}

 However, such a definition is not sufficiently strong in order that the
distributive property  of conjunction over disjunction is ensured. Hence, we need a stronger definition for conjunction.

\begin{definition}[Category-theoretical semantics of  conjunction]\label{def:conj}
 
Let $C, D, E, C \sqcap D, C \sqcap E$, $D \sqcup E$, $C\sqcap (D \sqcup E)$  and $(C \sqcap D)\sqcup  (C \sqcap E)$ be   objects of $\mathscr{C}_c\langle C_0, \mathcal{O}\rangle$. Category-theoretical semantics of $C \sqcap D$ is defined by using the following arrows in $\mathscr{C}_c\langle C_0, \mathcal{O}\rangle$. 
\begin{align}
&C \sqcap D\longrightarrow  C \text{ and } C \sqcap D  \longrightarrow  D   \label{conj01} \\
& \forall X, X  \longrightarrow  C \text{ and } X \longrightarrow  D    \Longrightarrow  X  \longrightarrow C \sqcap D \label{conj1}\\
& C \sqcap (D \sqcup E) \longrightarrow (C\sqcap D)\sqcup  (C\sqcap E)\label{conj2} 
\end{align} 
\end{definition}

 Note that under the set-theoretical semantics  the distributive property of disjunction over conjunction is not \emph{independent}, i.e. it is a consequence of the definitions of disjunction and conjunction. However, this does not hold under the category-theoretical semantics.   The following lemma provides the connection between the usual set semantics of conjunction and the category-theoretical one given in Definition~\ref{def:conj}. In this lemma, Properties~(\ref{conj001}-\ref{conj003}) 
  are rewritings of Arrows~(\ref{conj01}-\ref{conj2}) in set theory.  

\begin{lemma}\label{lem:conj}
The category-theorectical semantics of $C\sqcap D$ characterized by  Definition~\ref{def:conj} is compatible with the set-theoretical semantics of $C\sqcap D$, that means
if $\langle \Delta^\mathcal{I}, \cdot^\mathcal{I}\rangle$ is an interpretation under set-theoretical semantics, then the following holds:  

$(U\sqcap V)^\mathcal{I}=U^\mathcal{I}\cap V^\mathcal{I}$ for all concepts $U,V$  iff
\begin{align}
&(C \sqcap D)^\mathcal{I} \subseteq  C^\mathcal{I},    (C \sqcap D)^\mathcal{I} \subseteq  D^\mathcal{I}   \label{conj001} \\
& \forall X\subseteq \Delta^\mathcal{I}, X \subseteq  C^\mathcal{I}, X  \subseteq  D^\mathcal{I} \Longrightarrow 
X  \subseteq (C \sqcap D)^\mathcal{I}  \label{conj002}\\
&(C \sqcap (D \sqcup E))^\mathcal{I} \subseteq  ((C \sqcap D)\sqcup (C \sqcap E))^\mathcal{I}  \label{conj003}\\
&\hspace{1cm}\text{for all concepts } C,D \text{ and } E.\nonumber
\end{align}
\end{lemma}
\begin{proof}
 \noindent ``$\Longleftarrow$". Due to (\ref{conj001}) we have $(U\sqcap V)^\mathcal{I}  \subseteq U^\mathcal{I}\cap V^\mathcal{I}$. Let $X\subseteq \Delta^\mathcal{I}$   such that $X =U^\mathcal{I}\cap V^\mathcal{I}$. This implies that $X \subseteq U^\mathcal{I}$ and $X \subseteq V^\mathcal{I}$.   Due to (\ref{conj002}), we have  $X\subseteq (U\sqcap V)^\mathcal{I}$.  
 
 \noindent ``$\Longrightarrow$".  From $(C\sqcap D)^\mathcal{I}=C^\mathcal{I}\cap D^\mathcal{I}$   we obtain (\ref{conj001}).    Moreover, if   $X\subseteq C^\mathcal{I}$ and $X\subseteq D^\mathcal{I}$ then $X\subseteq C^\mathcal{I}\cap D^\mathcal{I}=(C\sqcap D)^\mathcal{I}$ by the hypothesis. Thus, (\ref{conj002}) is proved. To prove (\ref{conj003}), we use the hypothesis and the usual set-theoretical semantics as follows:  $(C  \sqcap (D\sqcup E))^\mathcal{I}=C^\mathcal{I} \cap (D\sqcup E)^\mathcal{I}=C^\mathcal{I} \cap (D^\mathcal{I} \cup E^\mathcal{I})=(C^\mathcal{I} \cap D^\mathcal{I}) \cup (C^\mathcal{I} \cap E^\mathcal{I})=(C \sqcap D)^\mathcal{I} \cup (C \sqcap E)^\mathcal{I}=((C \sqcap D)  \sqcup (C \sqcap E))^\mathcal{I}$. 
\end{proof}

 With disjunction and conjunction, one can use the arrows $C\sqcap \neg C\longrightarrow \bot$ and $\top \longrightarrow C\sqcup \neg C$ to define  category-theoretical semantics of  negation. However, such a definition does not allow to entail useful properties (c.f. Lemma~\ref{lem:neg-prop}) which are available    under the set-theoretical semantics.   Therefore, it is required to use more properties to characterize negation  under the category-theoretical semantics.
 
The following example describes a category which verifies all arrows from (\ref{disj1}) to (\ref{conj01}) and (\ref{conj1}), but not  (\ref{conj2}).  This shows that the weak definition of conjunction with Arrows~(\ref{conj01}) and (\ref{conj1}) does not allow to entail Arrow~(\ref{conj2}) (distributivity).

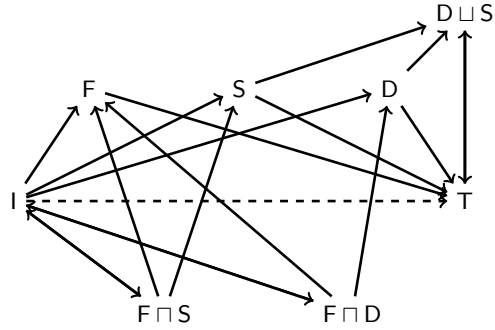
\begin{figure}[!htbp]
\centering
\begin{tikzpicture}[line width=1pt]  
\node[] (I) at (-3,-1) [label={[xshift=-0.3cm, yshift=-0.6cm]}] {$\mathsf{I}$};
\node[] (T) at (3,-1) [label={[xshift=-0.3cm, yshift=-0.6cm]}] {$\mathsf{T}$};
\node[] (X) at (-2,0.5) [label={[xshift=-0.3cm, yshift=-0.6cm]}] {$\mathsf{F}$};
\node[] (Y) at (0,0.5) [label={[xshift=-0.3cm, yshift=-0.6cm]}] {$\mathsf{S}$};
\node[] (Z) at (2,0.5) [label={[xshift=-0.3cm, yshift=-0.6cm]}] {$\mathsf{D}$};

\node[] (DS) at (3,1.5) [label={[xshift=-0.3cm, yshift=-0.6cm]}] {$\mathsf{D}\sqcup \mathsf{S}$};

\node[] (FD) at (1.5,-2.5) [label={[xshift=-0.3cm, yshift=-0.6cm]}] {$\mathsf{F}\sqcap \mathsf{D}$};

\node[] (FS) at (-1,-2.5) [label={[xshift=-0.3cm, yshift=-0.6cm]}] {$\mathsf{F}\sqcap \mathsf{S}$};

\draw[->] (Y) to[]node[left=2pt]{ } (DS);

\draw[->] (Z) to[]node[left=2pt]{ } (DS);

\draw[->] (DS) to[]node[right=2pt]{ } (T);

\draw[->] (T) to[]node[left=2pt]{ } (DS);

\draw[->] (FD) to[]node[left=2pt]{ } (Z);

\draw[->] (FD) to[]node[left=2pt]{ } (X);

\draw[->] (FS) to[]node[left=2pt]{ } (Y);

\draw[->] (FS) to[]node[left=2pt]{ } (X);

\draw[->] (FS) to[]node[left=2pt]{ } (I);

\draw[->] (I) to[]node[left=2pt]{ } (FS);

\draw[->] (FD) to[]node[left=2pt]{ } (I);

\draw[->] (I) to[]node[left=2pt]{ } (FD);
  
 \draw[->] (I) to[]node[left=2pt]{ } (X);
 \draw[->] (X) to[]node[right=2pt]{ } (T);
 \draw[->] (Y) to[]node[right=2pt]{ } (T);
 \draw[->] (Z) to[]node[right=2pt]{ } (T);
 \draw[->] (I) to[]node[left=2pt]{ } (Y); 
 \draw[->] (I) to[]node[left=2pt]{ } (Z); 
 \draw[->,dashed] (I) to[]node[left=2pt]{ } (T); 
  
\end{tikzpicture}
\caption{Category for Example~\ref{ex:distrib}. Some arrows resulting from  transitivity are omitted. }\label{fig:example2}
\end{figure}

\begin{example}\label{ex:distrib} For the sake of brevity, we rename the objects introduced in Example~\ref{ex:intro}  such as   $\mathsf{arrived}$, $\mathsf{filled\mhyphen  room}$, $\mathsf{starting}$, $\mathsf{started}$,
    $\mathsf{finished}$
to $\mathsf{I}$, $\mathsf{F}$,  $\mathsf{S}$,
$\mathsf{D}$,$\mathsf{T}$
respectively. Figure~\ref{fig:example2}  describes a category consisting of  the objects and arrows from Example~\ref{ex:intro} with new ones. In this category, the conjunctions and disjunction with the arrows satisfy   Definitions~\ref{def:disj} and \ref{def:conj}. We have  
 $\mathsf{D}\sqcup \mathsf{S}\longleftrightarrows \mathsf{T}$, and thus $\mathsf{F} \sqcap (\mathsf{D}\sqcup \mathsf{S}) \longleftrightarrows \mathsf{F}\sqcap \mathsf{T} \longleftrightarrows\mathsf{F}$. Moreover, we have $\mathsf{F}\sqcap \mathsf{D}\longleftrightarrows \mathsf{I}$, $\mathsf{F}\sqcap \mathsf{S}\longleftrightarrows \mathsf{I}$, and thus $(\mathsf{F}\sqcap \mathsf{D})\sqcup (\mathsf{F}\sqcap \mathsf{S}) \longleftrightarrows \mathsf{I}$.  
However, since $\mathsf{F}\not\longrightarrow \mathsf{I}$, the distributive property does not hold in this category. 
\end{example}

 With disjunction and conjunction, one can use the arrows $C\sqcap \neg C\longrightarrow \bot$ and $\top \longrightarrow C\sqcup \neg C$ to define  category-theoretical semantics of  negation. However, such a definition does not allow to entail useful properties (c.f. Lemma~\ref{lem:neg-prop}) which are available    under the set-theoretical semantics.   Therefore, it is required to use more properties to characterize negation  under category-theoretical semantics. Indeed, the following example shows that $C\longleftrightarrows \neg \neg C$ do not hold.

\begin{example}\label{ex:neg}
We consider the category  in Figure~\ref{fig:example3}. We have $\mathsf{F} \sqcap \mathsf{S}\longrightarrow \mathsf{I}$ and $\mathsf{T}\longrightarrow  \mathsf{F} \sqcup \mathsf{S}$.  Hence,  we can put $\neg \mathsf{F}=  \mathsf{S}$. Moreover,  we have $\mathsf{D} \sqcap \mathsf{S}\longrightarrow \mathsf{I}$ and $\mathsf{T}\longrightarrow  \mathsf{D} \sqcup \mathsf{S}$.  Hence,  we can put $\neg \mathsf{S}=  \mathsf{D}$, and thus $\neg \neg \mathsf{F}=  \mathsf{D}$. However, there does not exist the arrows $\mathsf{F}\longleftrightarrows  \mathsf{D}$ in this category.
\end{example}

\begin{figure}[!htbp]
\centering
\begin{tikzpicture}[line width=1pt]  
\node[] (I) at (-3,-1) [label={[xshift=-0.3cm, yshift=-0.6cm]}] {$\mathsf{I}$};
\node[] (T) at (3,-1) [label={[xshift=-0.3cm, yshift=-0.6cm]}] {$\mathsf{T}$};
\node[] (X) at (-2,0.5) [label={[xshift=-0.3cm, yshift=-0.6cm]}] {$\mathsf{F}$};
\node[] (Y) at (0,0.5) [label={[xshift=-0.3cm, yshift=-0.6cm]}] {$\mathsf{S}$};
\node[] (Z) at (2,0.5) [label={[xshift=-0.3cm, yshift=-0.6cm]}] {$\mathsf{D}$};

\node[] (DS) at (3,1.5) [label={[xshift=-0.3cm, yshift=-0.6cm]}] {$\mathsf{D}\sqcup \mathsf{S}$};

\node[] (FORS) at (3,-2.5) [label={[xshift=-0.3cm, yshift=-0.6cm]}] {$\mathsf{F}\sqcup \mathsf{S}$};

\node[] (DAS) at (0,-2.5) [label={[xshift=-0.3cm, yshift=-0.6cm]}] {$\mathsf{D}\sqcap \mathsf{S}$};

\node[] (FS) at (-3,-2.5) [label={[xshift=-0.3cm, yshift=-0.6cm]}] {$\mathsf{F}\sqcap \mathsf{S}$};

\draw[->] (Y) to[]node[left=2pt]{ } (DS);

\draw[->] (Z) to[]node[left=2pt]{ } (DS);

\draw[->] (DS) to[]node[right=2pt]{ } (T);

\draw[->] (T) to[]node[left=2pt]{ } (DS);

\draw[->] (T) to[]node[left=2pt]{ } (FORS);

\draw[->] (FORS) to[]node[left=2pt]{ } (T);

\draw[->] (Y) to[]node[left=2pt]{ } (FORS);

\draw[->] (X) to[]node[left=2pt]{ } (FORS);

\draw[->] (FS) to[]node[left=2pt]{ } (Y);

\draw[->] (FS) to[]node[left=2pt]{ } (X);

\draw[->] (FS) to[]node[left=2pt]{ } (I);

\draw[->] (I) to[]node[left=2pt]{ } (FS);

\draw[->] (DAS) to[]node[left=2pt]{ } (I);

\draw[->] (I) to[]node[left=2pt]{ } (DAS);

\draw[->] (DAS) to[]node[left=2pt]{ } (Y);
\draw[->] (DAS) to[]node[left=2pt]{ } (Z);
  
 \draw[->] (I) to[]node[left=2pt]{ } (X);
 \draw[->] (X) to[]node[right=2pt]{ } (T);
 \draw[->] (Y) to[]node[right=2pt]{ } (T);
 \draw[->] (Z) to[]node[right=2pt]{ } (T);
 \draw[->] (I) to[]node[left=2pt]{ } (Y); 
 \draw[->] (I) to[]node[left=2pt]{ } (Z); 
 \draw[->,dashed] (I) to[]node[left=2pt]{ } (T); 
  
\end{tikzpicture}
\caption{Category for Example~\ref{ex:neg}. Some arrows resulting from  transitivity are missing. }\label{fig:example3}
\end{figure}
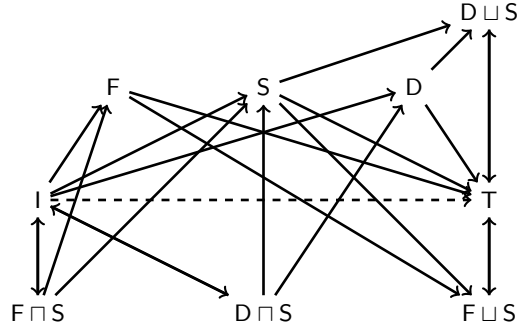

\begin{definition}[Category-theoretical semantics of  negation]\label{def:neg}

Let $C, \neg C, C \sqcap \neg C , C \sqcup \neg C$ be objects  of $\mathscr{C}_c$.  Category-theoretical semantics of $\neg C$ is defined by using the following arrows in $\mathscr{C}_c$. 
\begin{align}
& C \sqcap \neg C \longrightarrow \bot \label{neg-bot} \\
& \top\longrightarrow C \sqcup \neg C\label{neg-top}  \\
& \forall X, C \sqcap X \longrightarrow \bot \Longrightarrow X \longrightarrow \neg C \label{neg-max} \\
& \forall X, \top \longrightarrow C \sqcup X  \Longrightarrow \neg C \longrightarrow X \label{neg-min}
\end{align}
\end{definition}

\begin{lemma}\label{lem:neg}
The category-theorectical semantics of $\neg C$ characterized by Definition~\ref{def:neg} is compatible with the set-theoretical semantics of $\neg C$, that means if $\langle \Delta^\mathcal{I}, \cdot^\mathcal{I}\rangle$ is an interpretation      under set-theoretical semantics, then the following holds:  
 
  $C^\mathcal{I} \cap \neg C^\mathcal{I}\subseteq \bot^\mathcal{I}$ and $\top^\mathcal{I} \subseteq  C^\mathcal{I} \cup \neg C^\mathcal{I}$  imply
\begin{align}
& \forall X, (C \sqcap X)^\mathcal{I} \subseteq \bot^\mathcal{I} \Longrightarrow X^\mathcal{I} \subseteq (\neg C)^\mathcal{I} \label{neg3} \\
& \forall X, \top^\mathcal{I} \subseteq (C \sqcup X)^\mathcal{I}  \Longrightarrow (\neg C)^\mathcal{I} \subseteq X^\mathcal{I} \label{neg4}
\end{align}

\end{lemma}
Note that the properties  $C^\mathcal{I} \cap \neg C^\mathcal{I}\subseteq \bot^\mathcal{I}$ and $\top^\mathcal{I} \subseteq  C^\mathcal{I} \cup \neg C^\mathcal{I}$ suffice to characterize the semantics of negation  under the set-theoretical semantics.

\begin{proof}
 
\noindent 
Let $x\in X^\mathcal{I}$ with $(C\sqcap X)^\mathcal{I}\subseteq \bot^\mathcal{I}$.  Due to Lemma~\ref{lem:conj}, we have $(C\sqcap X)^\mathcal{I}=C^\mathcal{I} \cap X^\mathcal{I}$. Therefore, $x\notin C^\mathcal{I}$, and thus $x\in (\neg C)^\mathcal{I}$. Let $x\in (\neg C)^\mathcal{I}$ with $\top^\mathcal{I} \subseteq (C\sqcup X)^\mathcal{I} $. Hence, $x\notin C^\mathcal{I}$. Due to Lemma~\ref{lem:disj}, we have $(C\sqcup X)^\mathcal{I}=C^\mathcal{I} \cup X^\mathcal{I}$. Therefore,   $x\in X^\mathcal{I}$.\hfill$\square$
\end{proof}

Thanks to Properties~(\ref{neg-bot}-\ref{neg-min}), we obtain De Morgan's laws and other properties for category-theoretical semantics as follows.

\begin{lemma}\label{lem:neg-prop} The following arrows hold.
\begin{align}
&C \longleftrightarrows  \neg \neg C \label{neg-double}\\ 
&C\longrightarrow \neg D \Longrightarrow   D\longrightarrow \neg C \label{neg-dual}\\
&C \sqcap D \longrightarrow \bot \Longleftrightarrow C\longrightarrow \neg D, D\longrightarrow \neg C\label{neg-disjoint}\\
&\neg (C \sqcap D) \longleftrightarrows   \neg C \sqcup  \neg D \label{neg-conj}\\
&\neg (C \sqcup D) \longleftrightarrows   \neg C \sqcap  \neg D \label{neg-disj}
\end{align}
\end{lemma}

\begin{proof}
\begin{enumerate}[wide, labelwidth=!, labelindent=0pt]
\item By Properties~(\ref{neg-bot}) and (\ref{neg-max})  where $C$ gets $\neg C$ and $X$ gets $C$, we have $C\longrightarrow \neg \neg C$. Analogously,  by Properties~(\ref{neg-top}) and (\ref{neg-min}), we obtain  $\neg \neg C \longrightarrow C$.  Hence, (\ref{neg-double}) is proved. 
\item To prove (\ref{neg-dual}), we need $C\sqcap D\longrightarrow \bot$ which follows from $C\sqcap D \longrightarrow C\longrightarrow \neg D$ and $C\sqcap D \longrightarrow D$, and thus   $C\sqcap D \longrightarrow D\sqcap \neg D\longrightarrow \bot$.

\item To prove (\ref{neg-disjoint}), we start by proving    $C\longrightarrow \neg D  \Longrightarrow C \sqcap D \longrightarrow \bot$. We have $C\sqcap D \longrightarrow D$ and $C\sqcap D \longrightarrow C \longrightarrow \neg D$. By definition of conjunction, we obtain $C\sqcap D \longrightarrow D \sqcap \neg D\longrightarrow \bot$. To prove the other direction, we use (\ref{neg-max}) with $X=D$ or $X=C$.

\item To prove (\ref{neg-conj}), we need the definitions of conjunction, disjunction and negation. We have $C\sqcap D\longrightarrow C$ and $C\sqcap D\longrightarrow D$. Due to (\ref{neg-dual}), we obtain $\neg C \longrightarrow \neg(C\sqcap D)$ and $\neg D \longrightarrow \neg(C\sqcap D)$. By the definition of disjunction, we obtain $\neg C \sqcup \neg D \longrightarrow \neg(C \sqcap D)$. To prove the inverse, we take arrows $\neg C \longrightarrow \neg C  \sqcup  \neg D$ and $\neg D \longrightarrow \neg C  \sqcup  \neg D$ from the definition of disjunction. Due to (\ref{neg-dual}), it follows that $\neg (\neg C  \sqcup  \neg D) \longrightarrow C$ and $\neg (\neg C  \sqcup  \neg D) \longrightarrow D$. By Definition~\ref{def:conj}, we obtain $\neg(\neg C \sqcup \neg D) \longrightarrow  C \sqcap D$. Due to (\ref{neg-dual}) and (\ref{neg-double}), it follows that $\neg(C \sqcap D) \longrightarrow  \neg C \sqcup \neg D$.

\item Analogously, we can prove (\ref{neg-disj}) by starting with arrows  $\neg C  \sqcap  \neg D  \longrightarrow \neg C$ and $\neg C  \sqcap  \neg D \longrightarrow  \neg D$ from the definition of conjunction. Due to (\ref{neg-dual}), we have $C\longrightarrow \neg(\neg C  \sqcap  \neg D)$ and $D \longrightarrow  \neg(\neg C  \sqcap  \neg D)$. By the definition of disjunction, we have $C \sqcup D \longrightarrow \neg(\neg C  \sqcap  \neg D)$, and by (\ref{neg-dual}) we obtain  $(\neg C  \sqcap  \neg D) \longrightarrow \neg(C \sqcup D)$. To prove the inverse, we take arrows $C \longrightarrow C  \sqcup  D$ and $D \longrightarrow C  \sqcup  D$ obtained from the definition of disjunction. Due to  (\ref{neg-dual}), we have $\neg(C  \sqcup  D) \longrightarrow \neg C$ and $\neg(C  \sqcup  D) \longrightarrow  \neg D$. By Definition~\ref{def:conj}, we have $\neg(C  \sqcup  D) \longrightarrow \neg C \sqcap \neg D$. \hfill$\square$
\end{enumerate}
\end{proof}
In order to define category-theoretical semantics of existential restrictions, we need to introduce new objects and arrows  to $\mathscr{C}_c$ and  $\mathscr{C}_r$ as described in Figure~\ref{fig:exist}.

\begin{figure}[!htbp]
\centering
\begin{tikzpicture}[line width=1pt,scale=1] 
\node[] (C) at (0,-1) [label={[xshift=-0.3cm, yshift=-0.6cm]}] {$C$};
\node[] (codRP) at (-2,-1) [label={[xshift=-0.3cm, yshift=-0.6cm]}] {$\mathsf{cod}(R')$};
\node[] (codRE) at (2,-1) [label={[xshift=-0.3cm, yshift=-0.6cm]}] {$\mathsf{cod}(R_{(\exists R.C)})$};

\node[] (domRP) at (-2,-2.5) [label={[xshift=-0.3cm, yshift=-0.6cm]}] {$\mathsf{dom}(R')$};
\node[] (domRE) at (2,-2.5) [label={[xshift=-0.3cm, yshift=-0.6cm]}] {$\mathsf{dom}(R_{(\exists R.C)})$};

\node[] (RP) at (-2,1) [label={[xshift=-0.3cm, yshift=-0.6cm]}] {$R'$};
\node[] (RE) at (2,1) [label={[xshift=-0.3cm, yshift=-0.6cm]}] {$R_{(\exists R.C)}$};

\node[] (R) at (0,1) [label={[xshift=-0.3cm, yshift=-0.6cm]}] {$R$};

 \node[] (X) at (0,1) [label={[xshift=-0.3cm, yshift=-0.6cm]}] { };

 \draw[->] (codRP) to[]node[below=2pt]{${i}$} (C);
 
 \draw[->] (codRE) to[]node[below=2pt]{${j}$} (C);
 \draw[dashed,->] (RP) to[]node[right=2pt]{$\mathsf{cod}$} (codRP);
 
 \draw[dashed,->] (RE) to[]node[left=2pt]{$\mathsf{cod} $} (codRE);
 
 \draw[->] (RP) to[]node[above=2pt]{${k}$} (X); 
 \draw[->] (RE) to[]node[above=2pt]{${l}$} (X); 
  
 \path [->,dotted,out=210,in=135] (RP) edge node[left]{$\mathsf{dom}$} (domRP);
  \path [->,dotted,out=-30,in=45] (RE) edge node[right]{$\mathsf{dom}$} (domRE);
 \draw[->] (domRP) to[]node[above=2pt]{$m$} (domRE);
\end{tikzpicture}
\caption{Commutative diagram for existential restriction}\label{fig:exist}
\end{figure}
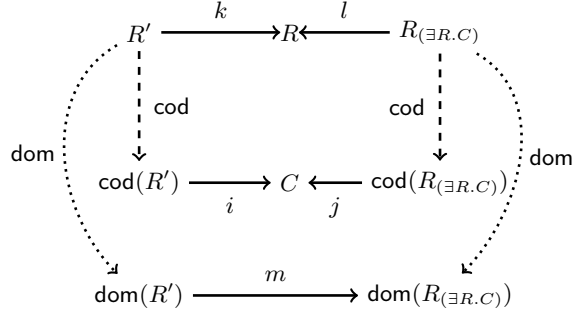

\begin{definition}[Category-theoretical semantics of  existential restriction]\label{def:exist}
Let $\exists R.C, C$ be objects of $\mathscr{C}_c$, and    
$R, R_{(\exists R.C)}$ be objects of  $\mathscr{C}_r$.
Category-theoretical semantics of $\exists R.C$ is defined by using arrows in $\mathscr{C}_c$ and $\mathscr{C}_r$ as follows. There is an arrow $j$ from $\mathsf{cod}(R_{(\exists R.C)})$ to $C$   in $\mathscr{C}_c$, an arrow $l$ from $R_{(\exists R.C)}$ to $R$ in $\mathscr{C}_r$, and if  there is an object  $R'$ of $\mathscr{C}_r$, an arrow $k$ from $R'$ to $R$ in  $\mathscr{C}_r$, and an arrow $i$ from $\mathsf{cod}(R')$ to $C$ in $\mathscr{C}_c$, then  there is an arrow $m$ in $\mathscr{C}_c$  such that the diagram in  Figure~\ref{fig:exist} commutes.

The diagram in Figure~\ref{fig:exist} can be rephrased  as follows:
\begin{align}
&R_{(\exists R.C)}  \longrightarrow R, \mathsf{cod}(R_{(\exists R.C)})  \longrightarrow C\label{ex-arrow01}\\ &\mathsf{dom}(R_{(\exists R.C)})\longleftrightarrows \exists R.C \label{ex-arrow02}\\
&\forall R', R'\longrightarrow R, \mathsf{cod}(R') \longrightarrow C\Longrightarrow  \mathsf{dom}(R')\longrightarrow \mathsf{dom}(R_{(\exists R.C)})  \label{ex-arrow03} 
\end{align}
\end{definition}

Since $\mathsf{dom}(R_{(\exists R.C)})\longleftrightarrows \exists R.C$ by (\ref{ex-arrow02}), the objects $\mathsf{dom}(R_{(\exists R.C)})$ and $\exists R.C$ are mutually  replaceable in the category context. 

\begin{lemma}\label{lem:exist}
The category-theorectical semantics of $\exists R.C$ characterized by Definition~\ref{def:exist} is compatible with the set-theoretical semantics of $\exists R.C$, that means
if $\langle \Delta^\mathcal{I}, \cdot^\mathcal{I}\rangle$ is an interpretation      under set-theoretical semantics such that  
\begin{align}
&R_{(\exists R.C)}^\mathcal{I} \subseteq R^\mathcal{I} \label{ex-arrow1}\\
&\mathsf{cod}(R_{(\exists R.C)})^\mathcal{I} \subseteq C^\mathcal{I} \label{ex-arrow2}
\end{align}
\noindent then the following holds:

  $\mathsf{dom}(R_{(\exists R.C)})^\mathcal{I}=\{x\in \Delta^\mathcal{I}\mid \exists y\in \Delta^\mathcal{I} : \langle x,y\rangle\in R^\mathcal{I} \wedge y\in C^\mathcal{I}\}$ iff
\begin{align}
&\forall R'\subseteq \Delta^\mathcal{I}\times \Delta^\mathcal{I}, {R'}\subseteq R^\mathcal{I}, \mathsf{cod}(R') \subseteq C^\mathcal{I}\Longrightarrow  \mathsf{dom}(R')\subseteq \mathsf{dom}(R_{(\exists R.C)})^\mathcal{I}   \label{ex-arrow3}
\end{align}
\end{lemma}

\begin{proof}
 \noindent ``$\Longleftarrow$". 
Let $x'\in \mathsf{dom}(R_{(\exists R.C)})^\mathcal{I}$.   There is an element $y\in   \mathsf{cod}(R_{(\exists R.C)})^\mathcal{I}$ such that $\langle x',y\rangle \in R_{(\exists R.C)}^\mathcal{I}$ by definition. Due to (\ref{ex-arrow2}),   we have $y\in   C^\mathcal{I}$. Analogously, due to (\ref{ex-arrow1}),     we have $\langle x',y\rangle \in R^\mathcal{I}$. Thus, $x'\in  \{x\in \Delta^\mathcal{I}\mid \exists y\in \Delta^\mathcal{I} : \langle x,y\rangle\in R^\mathcal{I} \wedge y\in C^\mathcal{I}\}$.

  Let $x'\in  \{x\in \Delta^\mathcal{I}\mid \exists y\in \Delta^\mathcal{I} : \langle x,y\rangle\in R^\mathcal{I} \wedge y\in C^\mathcal{I}\}$. By the set-theoretical semantics, there is an element $y\in   C^\mathcal{I}$ such that $\langle x',y\rangle \in R^\mathcal{I}$.  Take an $R'\subseteq \Delta^\mathcal{I} \times \Delta^\mathcal{I}$  such that $R'\subseteq R^\mathcal{I}$,  $\mathsf{cod}(R') \subseteq C^\mathcal{I}$ with   $x'\in \mathsf{dom}(R')$ and  $y\in \mathsf{cod}(R')$.   Due to (\ref{ex-arrow3}), we have $x'\in \mathsf{dom}(R_{(\exists R.C)})^\mathcal{I}$.
  
  \noindent ``$\Longrightarrow$". Let $R'\subseteq \Delta^\mathcal{I} \times \Delta^\mathcal{I}$   such that $R'\subseteq R^\mathcal{I}$, $\mathsf{cod}(R') \subseteq C^\mathcal{I}$. Let $x'\in \mathsf{dom}(R')$. This implies that there is some $y\in \mathsf{cod}(R')$ such that $\langle x',y\rangle\in  R'$. Hence, $y\in C^\mathcal{I}$ and $\langle x',y\rangle\in  R^\mathcal{I}$. Therefore, $x'\in  \{x\in \Delta^\mathcal{I}\mid \exists y\in \Delta^\mathcal{I} : \langle x,y\rangle\in R^\mathcal{I} \wedge y\in C^\mathcal{I}\}$ $=$ $\mathsf{dom}(R_{(\exists R.C)})^\mathcal{I}$.\hfill$\square$
\end{proof}

\begin{lemma}\label{lem:exist-prop} The following properties hold:
\begin{align}
&C \longrightarrow \bot \Longrightarrow \exists R.C  \longrightarrow \bot \label{exist-empty}\\
&C \longrightarrow D \Longrightarrow \exists R.C  \longrightarrow \exists R.D  \label{exist-sub}
\end{align}
\end{lemma}
\begin{proof}
\begin{enumerate}
    \item We have $\mathsf{cod}(R_{(\exists R.C)})\longrightarrow C$ by definition. By hypothesis $C\longrightarrow \bot$, we have $\mathsf{cod}(R_{(\exists R.C)})\longrightarrow \bot$ due Definition~\ref{def:syntax-cat} of $\mathsf{cod}$. Again,  due Definition~\ref{def:syntax-cat} of $\mathsf{cod}$, we have  $R_{(\exists R.C)} \longrightarrow R_\bot$. Moreover, we have $\mathsf{dom}(R_{(\exists R.C)})  \longrightarrow \bot$  due Definition~\ref{def:syntax-cat} of $\mathsf{dom}$. By Definition~\ref{def:exist}, we obtain $\exists R.C\longrightarrow \bot$.
    
    \item To prove (\ref{exist-sub}) we consider two objects $R_{(\exists R.C)}$ and $R_{(\exists R.D)}$ in $\mathscr{C}_r$. We have $R_{(\exists R.C)}\longrightarrow R$ and  $\mathsf{cod}(R_{(\exists R.C)})\longrightarrow C \longrightarrow D$. By definition, we obtain $\mathsf{dom}(R_{(\exists R.C)})\longrightarrow \mathsf{dom}(R_{(\exists R.D)})$.  
    \hfill$\square$
\end{enumerate}
\end{proof}

The set-theoretical semantics of universal restriction can be defined as  negation of existential restriction. However, the definition $\forall R.C \longleftrightarrows  \neg  \exists R.\neg C$ does not allow to obtain  usual connections between existential and universal restrictions such as $\exists R.D \sqcap \forall R.C \longrightarrow  \exists R.(D \sqcap  C)$  under  category-theoretical semantics.   Therefore, we need more arrows to define category-theoretical semantics of universal restriction as follows.

\begin{definition}[Category-theoretical semantics of universal restriction]\label{def:forall}
Let $\forall R.C$, $C$  be objects of $\mathscr{C}_c$,  and    
$R$ be an object of  $\mathscr{C}_r$.
   Category-theoretical semantics of $\forall R.C$ is defined by using arrows in $\mathscr{C}_c$ and $\mathscr{C}_r$ as follows.
\begin{align}
&\forall R.C \longleftrightarrows  \neg  \exists R.\neg C\label{all-arrow1} \\
&\forall R', R'\longrightarrow R, \mathsf{dom}(R') \longrightarrow \forall R.C \Longrightarrow \mathsf{cod}(R')\longrightarrow C \label{all-arrow2}  
\end{align}
\end{definition}

\begin{lemma}\label{lem:forall}
The category-theoretical semantics of $\forall R.C$ characterized by Definition~\ref{def:forall}  is compatible with the set-theoretical semantics of $\forall R.C$, that means
if $\langle \Delta^\mathcal{I}, \cdot^\mathcal{I}\rangle$   is an interpretation   under set-theoretical semantics, then the following holds:   
  
  $\forall R.C^\mathcal{I}=\{x\in \Delta^\mathcal{I}\mid  \langle x,y\rangle\in R^\mathcal{I}\Longrightarrow y\in C^\mathcal{I}\}$ iff
\begin{align}
&\forall R.C^\mathcal{I}= (\neg  \exists R.\neg C)^\mathcal{I} \label{all-arrow4}\\
&\forall R'\subseteq \Delta^\mathcal{I}\times \Delta^\mathcal{I}, R' \subseteq R^\mathcal{I}, \mathsf{dom}(R')  \subseteq  \forall R.C^\mathcal{I}   \Longrightarrow  \mathsf{cod}(R')^\mathcal{I}\subseteq C^\mathcal{I}  \label{all-arrow5}
\end{align}
\end{lemma}
\begin{proof}
 \noindent ``$\Longleftarrow$". We have $\{x\in \Delta^\mathcal{I}\mid  \langle x,y\rangle\in R^\mathcal{I}\Longrightarrow y\in C^\mathcal{I}\}$ is the complement of  $\{x\in \Delta^\mathcal{I}\mid \exists y\in \Delta^\mathcal{I} : \langle x,y\rangle\in R^\mathcal{I}\wedge y\in \neg C^\mathcal{I}\} = \exists R.\neg C^\mathcal{I}$. Hence, we have  $\forall R.C^\mathcal{I}=\{x\in \Delta^\mathcal{I}\mid  \langle x,y\rangle\in R^\mathcal{I}\Longrightarrow y\in C^\mathcal{I}\}$ from (\ref{all-arrow4}).
 
 \noindent ``$\Longrightarrow$". Let  $R'\subseteq R^\mathcal{I}$ with $\mathsf{dom}(R') \subseteq \forall R.C^\mathcal{I}$. Let $y\in \mathsf{cod}(R')$. There is some $x'\in \mathsf{dom}(R')$ such that $\langle x', y\rangle \in  {R'}^\mathcal{I}$. Since $R'\subseteq R^\mathcal{I}$ and  $\mathsf{dom}(R') \subseteq \forall R.C^\mathcal{I}$, we have $x\in \forall R.C^\mathcal{I}$ and $\langle x', y\rangle \in  {R}^\mathcal{I}$. By hypothesis, $y\in C^\mathcal{I}$.\hfill$\square$
\end{proof}

\begin{lemma}\label{lem:forall-prop} The following properties hold. 
\begin{align}
&\forall R.C \sqcap \exists  R.\neg C \longrightarrow  \bot \label{forall-exists-bot}\\
&C \longrightarrow D \Longrightarrow \forall R.C \longrightarrow  \forall R.D  \label{forall-sub}\\
&\exists R.D \sqcap \forall R.C \longrightarrow  \exists R.(D \sqcap  C) \label{forall-exist}
\end{align}
\end{lemma}
 
\begin{proof}
\begin{enumerate}[wide, labelwidth=!, labelindent=0pt]
\item Arrow (\ref{forall-exists-bot}) is a consequence of (\ref{all-arrow1}) and (\ref{neg-bot}).

    \item Due to (\ref{neg-dual}), $C\longrightarrow D$ implies $\neg D \longrightarrow \neg C$. By (\ref{exist-sub}), we obtain $\exists R.\neg D\longrightarrow \exists R.\neg C$. Again, due to (\ref{neg-dual}), we have $\neg \exists R.\neg C\longrightarrow \neg \exists R.\neg D$. By definition with (\ref{all-arrow1}), it follows $\forall R.C \longrightarrow \forall R.D$. 
    
    
    \item To prove (\ref{forall-exist}) we need category-theoretical semantics of existential and universal restrictions. We define an object $R_X$ in $\mathscr{C}_r$ such that $R_X\longrightarrow R_{(\exists R.D)}$ and  $\mathsf{dom}(R_X)\longleftrightarrows \exists R.D \sqcap \forall R.C$. This implies that $\mathsf{cod}(R_X)\longrightarrow  \mathsf{cod}(R_{(\exists R.D)})$ due to the functor $\mathsf{cod}$ from $\mathscr{C}_r$ to $\mathscr{C}_c$, and  $\mathsf{dom}(R_X)\longrightarrow   \forall R.C$. By definition, we have $\mathsf{cod}(R_{(\exists R.D)})\longrightarrow D$. Thus, $\mathsf{cod}(R_X)\longrightarrow D$.    Due to Arrow~(\ref{all-arrow2}) and $\mathsf{dom}(R_X)\longrightarrow   \forall R.C$, we obtain $\mathsf{cod}(R_X)\longrightarrow  C$. Hence, $\mathsf{cod}(R_X)\longrightarrow  C\sqcap D$. By the definition of existential restrictions, we obtain $\mathsf{dom}(R_X)\longrightarrow   \exists  R.(D \sqcap C)$.  $\hfill \square$
\end{enumerate}
\end{proof}



Note that both $\mathscr{C}_c\langle C_0, \mathcal{O}\rangle$ and $\mathscr{C}_r\langle C_0, \mathcal{O}\rangle$ may consist of more  objects. However, new arrows should be entailed by those existing or by using the properties given in Definitions~(\ref{def:disj}-\ref{def:forall}). Adding to $\mathscr{C}_c\langle C_0, \mathcal{O}\rangle$ a new arrow that is \emph{independent} from those existing leads to a semantic change of the ontology.   Since all properties in Lemma~ \ref{lem:neg-prop}, \ref{lem:exist-prop} and \ref{lem:forall-prop} are consequences of those given in these definitions, they can be used to obtain \emph{entailed}  arrows (i.e. not independent ones).

\begin{theorem}[Arrow and subsumption]\label{thm:arrow-entail} Let $C_0$ be an $\mathcal{ALC}$ concept and   $\mathcal{O}$ an $\mathcal{ALC}$ ontology. Let  $\mathscr{C}_c\langle C_0,\mathcal{O} \rangle$ be an ontology  category. It holds that $\langle C_0,\mathcal{O}\rangle \models X\sqsubseteq Y$ (under set-theoretical semantics) if $X\longrightarrow Y$ is an arrow in $\mathscr{C}_c\langle C_0, \mathcal{O}\rangle$.
\end{theorem}

\begin{proof} First, each  arrow  $\top\longrightarrow \neg E \sqcup F$  added to  $\mathscr{C}_c\langle C_0, \mathcal{O}\rangle$  implies $E \longrightarrow   F$ due to Arrows~(\ref{neg-min}) and (\ref{neg-double}) for each axiom $E\sqsubseteq F$ of $\mathcal{O}$. By set-theoretical semantics, we have  $\langle C_0,\mathcal{O}\rangle  \models E\sqsubseteq F$. Due to Lemmas~\ref{lem:disj}, \ref{lem:conj}, \ref{lem:neg}, \ref{lem:exist} and \ref{lem:forall}, for each arrow $X\longrightarrow Y$ introduced  by Definitions~\ref{def:disj}, \ref{def:conj}, \ref{def:neg}, \ref{def:exist},  \ref{def:forall}, 
  we have $\langle C_0,\mathcal{O}\rangle  \models X\sqsubseteq Y$.
\end{proof}

 We now  introduce  category-theoretical satisfiability of an  $\mathcal{ALC}$ concept  with respect to an $\mathcal{ALC}$  ontology.

\begin{definition}\label{def:cat-satisfiability}
Let $C_0$ be an $\mathcal{ALC}$ concept, $\mathcal{O}$ an $\mathcal{ALC}$ ontology.  $C$ is category-theoretically unsatifiable with respect to $\mathcal{O}$ if  there is an ontology category $\mathscr{C}_c\langle C_0, \mathcal{O}\rangle$ which  consists of the arrow $C_0\longrightarrow \bot$.
\end{definition}

Since an ontology category $\mathscr{C}_c\langle C_0, \mathcal{O}\rangle$  may consist of objects arbitrarily built from the signature, Definition~\ref{def:cat-satisfiability} offers possibilities to build a larger ontology category  from which  new arrows can be discovered by applying arrows given in Definitions~(\ref{def:disj}-\ref{def:forall}).

\begin{theorem}\label{thm:cat-set} Let $\mathcal{O}$  be an $\mathcal{ALC}$ ontology and $C$ an  $\mathcal{ALC}$ concept. $C$ is category-theoreti\-cally unsatifiable with respect to $\mathcal{O}$ iff $C$ is set-theoretically unsatifiable.
\end{theorem}
To prove this theorem, we need the following preliminary result.

\begin{lemma}\label{lem:set2cat}
Let $\mathcal{O}$  be an $\mathcal{ALC}$ ontology and $C$  an  $\mathcal{ALC}$ concept. If $C$ is set-theoretically unsatifiable with respect to $\mathcal{O}$ then $C$ is  category-theoretically unsatifiable. 
\end{lemma}

In order to prove this lemma, we use a tableau algorithm to generate from a unsatisfiable  $\mathcal{ALC}$ concept with respect to an ontology a set of completion trees each of which contains a clash ($\bot$ or a pair $\{A,\neg A\}$ where $A$ is a concept name).  To ensure self-containedness of the paper, we describe   here necessary elements which allow to follow the proof of the lemma. We refer the readers to \cite{baa2000,hor07} for formal details.     

We use $\mathsf{sub}(C,\mathcal{O})$ to denote  a set of subconcepts in NNF (i.e. negations appear only in front of concept name) from $C$ and $\mathcal{O}$. This set can be built by applying the following rules :(i) $C \in\mathsf{sub}(C,\mathcal{O})$ and $\neg E \sqcup F \in\mathsf{sub}(C,\mathcal{O})$ for each axiom $E\sqsubseteq F$; (ii) if $\neg C \in\mathsf{sub}(C,\mathcal{O})$ then $C\in \mathsf{sub}(C,\mathcal{O})$; (iii)  if $C \sqcap D$ or  $C \sqcup D\in\mathsf{sub}(C,\mathcal{O})$ then $C, D\in \mathsf{sub}(C,\mathcal{O})$; (iv) if $\exists R.C $ or  $\forall R.C \in\mathsf{sub}(C,\mathcal{O})$ then $C\in \mathsf{sub}(C,\mathcal{O})$. A \emph{completion tree} $T=\langle v_0, V, E, L\rangle$ is a tree where $V$ is a set of nodes, and each node $x\in V$ is labelled with  $L(x)\subseteq \mathsf{sub}(C,\mathcal{O})$, and  a root $v_0\in V$ with $C\in L(v_0)$; $E$ is a set of edges, and each edge $\langle x,y\rangle\in E$ is labelled with a role $L\langle x,y\rangle=\{R\}$. In a completion tree $T$, a node $x$ is \emph{blocked} by an ancestor $y$ if $L(x)=L(y)$. 

To build completion trees, a tableau algorithm starts by initializing a tree $T_0$ and  applies the following \emph{completion rules} to each \emph{clash-free} tree $T$ : \textbf{[$\sqsubseteq$-rule]} for each axiom $E\sqsubseteq F$ and each node $x$, $\neg E\sqcup F \in L(x)$; \textbf{[$\sqcap$-rule]} if $E \sqcap F\in L(x)$ then $E,F\in L(x)$; \textbf{[$\sqcup$-rule]} if $E \sqcup F\in L(x)$ and $\{E,F\}\cap L(x)=\emptyset$ then it creates two copies $T_1$ and $T_2$ of the current tree $T$,  and set $L(x_1)\gets L(x_1)\cup \{E\}$, $L(x_2)\gets L(x_2)\cup \{F\}$ where $x_i$ in  $T_i$   is a copy of $x$ from $T$.  \textbf{[$\exists$-rule]} if $\exists R.C\in L(x)$, $x$ is not blocked and $x$ has no  edge $\langle x,x'\rangle$ with $L(\langle x,x'\rangle=\{R\}$ then it creates a successor $x'$ of $x$, and set $L(x')\gets\{C\}$, $L(\langle x,x'\rangle\gets\{R\}$; \textbf{[$\forall$-rule]} if $\forall R.C\in L(x)$ and $x$ has a successor $x'$ with    $L(\langle x,x'\rangle=\{R\}$ then it adds $C$ to $L(x')$.

When   \textbf{[$\sqcup$-rule]} is applied   to a node $x$ of a completion tree $T$ with $E\sqcup F\in L(x)$, it   generates two children trees $T_1$ and $T_2$  with a node $x_1$ in $T_1$ such that $E\sqcup F,E\in L(x_1)$,   and a node $x_2$ in $T_2$ such that $E\sqcup F,F\in L(x_2)$ as described above.  In this case,  we say that $T$ is parent of $T_1$ and $ T_2$ by $x$; or $T_1$ and $ T_2$ are children of $T$ by $x$;   
$x$ is called a \emph{disjunction} node by $E\sqcup F$; and $x_1,x_2$ are called  \emph{disjunct} nodes of $x$ by $E\sqcup F$.  We use $\mathbb{T}$ to denote the tree whose nodes are completion trees generated by the tableau algorithm as described above.  

A completion tree is \emph{complete} if no completion rule is applicable. It was shown that if $C$ is set-theoretically unsatisfiable with respect to $\mathcal{O}$ then all \emph{complete completion trees} contain a clash (an incomplete completion tree may contain a clash)  \cite{baa2000,hor07}. Since this result does not depend on the order of applying completion rules to nodes, we can assume in this proof that the following order   \textbf{[$\sqsubseteq$-rule]}, \textbf{[$\sqcap$-rule]}, \textbf{[$\forall$-rule]}, \textbf{[$\exists$-rule]} and  \textbf{[$\sqcup$-rule]} is used, and a completion rule should be applied to the most ancestor node if applicable. 
Some properties are drawn from this assumption.
\begin{enumerate}[label=(P\arabic*), leftmargin=*]
\item\label{prop:1} If a node $x$ of a completion tree $T$ contains a clash then $x$ is a leaf node of $T$.

\item\label{prop:2} For a completion tree $T$, if  a node $y$ is an ancestor of a node $x$  in  $T$ such that $x,y$ are disjunction nodes, then the children trees by $y$ are ancestors of the children trees by $x$   in $\mathbb{T}$.

\end{enumerate}

\noindent \textbf{Proof of Lemma~\ref{lem:set2cat}}. We define an ontology category $\mathscr{C}_c\langle C, \mathcal{O}\rangle$  from $\mathbb{T}$ by starting from leaf trees of   $\mathbb{T}$. By \ref{prop:1}, each leaf tree $T$ of  $\mathbb{T}$ contains a clash in a leaf node $y$.  We have  $\bigsqcap_{Y\in L(y)} Y\longrightarrow \bot$. We add an object $\bigsqcap_{Y\in L(y)} Y$ to $\mathscr{C}_c\langle C, \mathcal{O}\rangle$ with the arrow. In the sequel, we try to define a sequence of arrows started with $\bigsqcap_{Y\in L(y)} Y\longrightarrow \bot$  which makes clashes propagate from the leaves into the root of $\mathbb{T}$. This propagation has to get through two crucial kinds of passing: from a node of a completion tree to an ancestor that is a disjunct node; and from such a disjunct node in a completion tree to its disjunction node in the parent completion tree in $\mathbb{T}$.

Let $T_1$ and $T_2$ be two leaf trees of $\mathbb{T}$ whose parent is $T$, and $y_1, y_2$ two leaf nodes of $T_1,T_2$ containing a clash.  By construction, $T_i$ has  a disjunct node $x_i$  which is an ancestor of $y_i$     with  $x_i\neq y_i$, and there is no disjunct node between $x_i$ and $y_i$ (among descendants of $x_i$). For each node $z$ between $y_i$ and $x_i$ if it exists, we define an object $\bigsqcap_{Z\in L(z)} Z$ and add them to $\mathscr{C}_c\langle C, \mathcal{O}\rangle$.  Let $z'$ be the parent node of $y_i$. There is  a concept $\exists R.D\in L(z')$ and $D\in L(y_i)$. In addition, if  $\forall R.D_1\in L(z')$ and \textbf{[$\forall$-rule]} is applied to $z'$ then  $D_1\in L(y_i)$. We show that $D\longrightarrow Z$ or  $D_1\longrightarrow Z$ for each concept $Z\in L(y_i)$. By construction,    \textbf{[$\sqcup$-rule]} is not applied to $y_i$ (and any node from $x_i$ to $y_i$). If $Z=\neg E \sqcup F$ comes from an axiom $E\sqsubseteq F$, then $D\longrightarrow \top \longrightarrow Z$. If $Z$ comes from  some $Z \sqcap Z'$  then we have to have already  $D\longrightarrow  Z \sqcap Z'\longrightarrow Z$ or $D_1\longrightarrow  Z \sqcap Z'\longrightarrow Z$. Hence, $\bigsqcap_{Y\in L(y_i)} Y\longrightarrow \bot$ implies $D\sqcap \bigsqcap_{\forall R.D_i\in L(z')} D_i\longrightarrow \bot$. By (\ref{exist-empty}) and  (\ref{forall-exist}), we obtain $\exists R.D\sqcap \bigsqcap_{\forall R.D_i\in L(z')} \forall R.D_i\longrightarrow \bot$, and thus $\bigsqcap_{Z\in L(z')} Z\longrightarrow \bot$. By using the same argument, we obtain $\bigsqcap_{X\in L(x_i)} X\longrightarrow \bot$. We add an object  $\bigsqcap_{X\in L(x_i)} X$ to $\mathscr{C}_c\langle C, \mathcal{O}\rangle$ with the arrow.

We now show $\bigsqcap_{X\in L(x)} X\longrightarrow \bot$ where $x$ is the disjunction node of $x_i$, $L(x)=W\cup \{E\sqcup F\}$,  $L(x_1)=W\cup \{E\sqcup F, E\}$, $L(x_2)=W\cup \{ E\sqcup F, F\}$. Due to (\ref{conj2}), we have $\bigsqcap_{V\in W} V\sqcap (E\sqcup F) \longrightarrow (\bigsqcap_{V\in W} V\sqcap  E) \sqcup (\bigsqcap_{V\in W} V \sqcap F)$. Moreover, since $\bigsqcap_{X\in L(x_i)} X\longrightarrow \bot$, we have $ (\bigsqcap_{V\in W} V\sqcap  E) \longrightarrow (\bigsqcap_{V\in W} V\sqcap (E\sqcup F) \sqcap E \longrightarrow  \bigsqcap_{X\in L(x_1)} X \longrightarrow\bot$ (we have $(E\sqcup F)\sqcap E\longrightarrow E$ due to (\ref{disj01}) and (\ref{conj1})), and $(\bigsqcap_{V\in W} V\sqcap  F) \longrightarrow (\bigsqcap_{V\in W} V\sqcap (E\sqcup F) \sqcap F \longrightarrow  \bigsqcap_{X\in L(x_2)} X \longrightarrow\bot$. Therefore, we obtain $\bigsqcap_{V\in W} V\sqcap (E\sqcup F) \longrightarrow\bot$ due to (\ref{conj2}), then add an object $\bigsqcap_{V\in W} V\sqcap (E\sqcup F)$ to $\mathscr{C}_c\langle C, \mathcal{O}\rangle$ with the arrow.

We can apply the same argument from $x$ to the next disjunct node  in $T$ which is an ancestor of $x$ according to \ref{prop:2},  and go upwards in $\mathbb{T}$ to find its parent and sibling. This process can continue and reach the root tree $T_0$ of $\mathbb{T}$. We get   $\bigsqcap_{X\in L(x_0)} X \longrightarrow \bot$ where $x_0$ is the root node of $T_0$. By construction, we have $\{C\}\cup \{\neg E\sqcup F\mid E\sqsubseteq F\in \mathcal{O}\}\subseteq L(x_0)$.  Let $X\in L(x_0)$. If $X=\neg E \sqcup F$ for some axiom $E  \sqsubseteq F\in \mathcal{O}$ then  $C\longrightarrow \top \longrightarrow \neg E \sqcup F$. Moreover, if $X=E$ or $X=F$ with $C= E \sqcap F$ then $C\longrightarrow E \sqcap F$. This implies that    $C  \longrightarrow  \bigsqcap_{X\in L(x_0)} X \longrightarrow \bot$. Hence,   $C \longrightarrow  \bot$. It is straightforward that $\mathscr{C}_c\langle C, \mathcal{O}\rangle$ is an ontology category consisting of  $C \longrightarrow  \bot$. Moreover, an arrow $\top \longrightarrow\neg E\sqcap F$ is added to $\mathscr{C}_c\langle C, \mathcal{O}\rangle$ for each axiom $E\sqsubseteq F$. This completes the proof of the lemma.

\noindent \textbf{Proof of Theorem~\ref{thm:cat-set}}.
\noindent ``$\Longleftarrow$". Let $\mathscr{C}_c\langle C, \mathcal{O}\rangle$ be an ontology category. Assume that there is   an arrow $C\longrightarrow \bot$ in $\mathscr{C}_c\langle C, \mathcal{O}\rangle$. Every
 arrow $X\longrightarrow Y$ added to $\mathscr{C}_c\langle C, \mathcal{O}\rangle$  must be one of following cases:
     (i) $X\sqsubseteq Y$ is an axiom of $\mathcal{O}$. Thus, $\mathcal{O}\models X\sqsubseteq Y$.
     (ii) $X\longrightarrow Y$ is added by Definitions~\ref{def:disj}, \ref{def:conj}, \ref{def:neg}, \ref{def:exist},  \ref{def:forall}. Due to  Lemma~\ref{lem:arrow-entail}, we have $\mathcal{O}\models X\sqsubseteq Y$.
     (iii) $X\longrightarrow Y$ is obtained by transitivity from  $X\longrightarrow Z$ and $Z\longrightarrow Y$. It holds that    $\mathcal{O}\models X\sqsubseteq Z$ and $\mathcal{O}\models Z\sqsubseteq Y$ imply $\mathcal{O}\models X\sqsubseteq Y$.
Hence, if $\mathscr{C}_c\langle C, \mathcal{O}\rangle$ consists of   an arrow $C\longrightarrow \bot$, then $\mathcal{O}\models C\sqsubseteq \bot$.

\noindent ``$\Longrightarrow$". A consequence of Lemma~\ref{lem:set2cat}.\hfill$\square$

\section{Conclusion and Future Work}

We have presented a rewriting of the usual set-theoretical semantics of $\mathcal{ALC}$  by using categorial language and showed that these two semantics are equivalent in the sense that the problem of concept unsatisfiability is the same in both settings. Thanks to the modular representation of  category-theoretical semantics  composed of separate constraints, one can identify and design interesting sublogics by removing some arrows from categorical definitions rather than an entire constructor. Such sublogics  may necessitate new reasoning procedures. We believe that category-theoretical semantics can be extended to more expressive DLs with role constructors. For instance, role functionality can be expressed as \emph{monic and epic arrows}.

\bibliographystyle{elsarticle-num} \bibliography{categoryALC}

\begin{thebibliography}{1}
\expandafter\ifx\csname url\endcsname\relax
  \def\url#1{\texttt{#1}}\fi
\expandafter\ifx\csname urlprefix\endcsname\relax\def\urlprefix{URL }\fi
\expandafter\ifx\csname href\endcsname\relax
  \def\href#1#2{#2} \def\path#1{#1}\fi

\bibitem{baa10}
F.~Baader, D.~Calvanese, D.~L. McGuinness, D.~Nardi, P.~F. Patel-Schneider
  (Eds.), The Description Logic Handbook: Theory, Implementation and
  Applications, Second Edition, Cambridge University Press, 2010.

\bibitem{pat04}
P.~Patel-Schneider, P.~Hayes, I.~Horrocks, Owl web ontology language semantics
  and abstract syntax, in: W3C Recommendation, 2004.

\bibitem{grau08}
B.~{Cuenca Grau}, I.~Horrocks, B.~Motik, B.~Parsia, P.~Patel-Schneider,
  U.~Sattler, Owl 2: The next step for owl, journal of web semantics: Science,
  services and agents, World Wide Web 6 (2008) 309–322.

\bibitem{law64}
W.~Lawvere, An elementary theory of the category of sets, in: Proc. Nat. Acad.
  Sci. (USA), 1964, p. 1506–1511.

\bibitem{gol06}
R.~Goldblatt, Topoi - The categorial analisys of logic, Mathematics, Dover
  publications, 2006.

\bibitem{spivak2012}
D.~I. Spivak, R.~E. Kent, Ologs: a categorical framework for knowledge
  representation, PloS one 7~(1) (2012) e24274.

\bibitem{saunders92}
S.~{Mac Lane}, I.~Moerdijk, Sheaves in Geometry and Logic, Springer, 1992.

\bibitem{baa2000}
F.~Baader, U.~Sattler, Tableau algorithms for description logics, in:
  Proceedings of the International Conference on Automated Reasoning with
  Tableaux and Related Methods, Vol. 1847, Springer-Verlag, St Andrews,
  Scotland, UK, 2000, p. 118.

\bibitem{hor07}
I.~Horrocks, U.~Sattler, A tableau decision procedure for $\mathcal{SHOIQ}$,
  Journal Of Automated Reasoning 39~(3) (2007) 249--276.

\end{thebibliography}

\end{document}